\documentclass[11pt]{article}
\usepackage{fullpage,times,url,natbib}
\usepackage{amsthm,amsfonts,amsmath,amssymb,epsfig,color,float,graphicx,verbatim,mathtools}
\usepackage{algorithm,algpseudocode}
\usepackage{enumitem}

\usepackage{hyperref}
\usepackage{tablefootnote}
\hypersetup{
	colorlinks   = true, 
	urlcolor     = blue, 
	linkcolor    = blue, 
	citecolor   = blue 
}
\usepackage{mathtools}
\usepackage{multirow}
\DeclarePairedDelimiter\ceil{\lceil}{\rceil}

\newtheorem{lemma}{Lemma}[]
\newtheorem{theorem}{Theorem}[]

\newtheorem{definition}{Definition}[]
\newtheorem{remark}{Remark}[]

\newcommand{\R}{\mathbb{R}}

\newcommand{\Z}{\mathbb{Z}}
\newcommand{\A}{\mathcal{A}}
\newcommand{\B}{\mathcal{B}}
\newcommand{\X}{\mathcal{X}}
\newcommand{\Y}{\mathcal{Y}}

\newcommand{\F}{\mathcal{F}}
\newcommand\norm[1]{\left\lVert#1\right\rVert}
\newcommand{\comp}[1]{\overline{#1}}

\title{Are We Still Missing an Item?}
\author{
  Roey Magen\footnote{Department of Computer Science and Applied Mathematics, Weizmann Institute of Science, Rehovot, Israel. Email: \href{mailto:roey.magen@weizmann.ac.il}{roey.magen@weizmann.ac.il}.
  }
}

\begin{document}
\maketitle
\begin{abstract}
The missing item problem, as introduced by Stoeckl in his work at SODA 23, focuses on continually identifying a missing element $e$ in a stream of elements ${e_1, ..., e_{\ell}}$ from the set $\{1,2,...,n\}$, such that $e \neq e_i$ for any $i \in \{1,...,\ell\}$. Stoeckl's investigation primarily delves into scenarios with $\ell<n$, providing bounds for the (i) deterministic case, (ii) the static case- where the algorithm might be randomized but the stream is fixed in advanced and (iii) the adversarially robust case- where the algorithm is randomized and each stream element can be chosen depending on earlier algorithm outputs. Building upon this foundation, our paper addresses previously unexplored aspects of the missing item problem.

In the first segment, we examine the static setting with a long stream, where the length of the steam $\ell$ is close to or even exceeds the size of the universe $n$. We present an algorithm demonstrating that even when $\ell$ is very close to $n$ (say $\ell=n-1$), polylog($n$) bits of memory suffice to identify the missing item. When the stream’s length $\ell$ exceeds the size of the universe $n$ i.e. $\ell = n +k$, we show a tight bound of roughly $\Tilde{\Theta}(k)$.

The second segment focuses on the {\em adversarially robust setting}. We show a lower bound for a pseudo-deterministic error-zero (where the algorithm reports its errors) algorithm of approximating $\Omega(\ell)$, up to polylog factors. Based on Stoeckl’s work and the previous result, we establish a tight bound for a random-start (only use randomness at initialization) error-zero streaming algorithm of roughly $\Theta(\sqrt{\ell})$.

In the final segment, we explore streaming algorithms with randomness-on-the-fly, where the random bits that are saved for future use are included in the space cost. For streams with length $\ell = O(\sqrt{n})$, we provide an upper bound of $O(\log n)$. This establishes a gap between randomness-on-the-fly to zero-error random-start.

\end{abstract}
\section{Introduction}
Streaming algorithms act on flows of massive data that arrive rapidly and cannot be stored due to a lack of memory (or lack of bandwidth to the memory).
Such algorithms are applied across a wide spectrum of domains, from network monitoring and financial analytics to recommendation systems and online advertising.
A typical assumption when designing and analyzing streaming algorithms is that the entire
stream is fixed in advance and is just provided to the streaming algorithm one item at a time, or
at least that the choice of the items in the stream is independent of the random bits of the streaming algorithm.  

Recently, there has been a growing interest in streaming algorithms that maintain correctness even when the choice of
stream items depends on previous answers given by the streaming algorithm and can hence may depend
on the internal state and the random bits of the algorithm. Such streaming algorithms are said to be adversarially robust. One way to tackle such an adversarial setting is the notion of pseudo-determinism: A pseudo-deterministic algorithm is a (randomized) algorithm that, when run multiple times on the same input, each time an independent choice of randomness, outputs the same output on all executions with high probability over its random coins. A pseudo-deterministic algorithm is necessarily adversarially robust since its outputs (with high probability) do not reveal any information about its internal state.  Another interesting question is whether large amounts of random bits give any advantage in streaming tasks. In other words, does the space complexity change when the random bits used by the streaming algorithm are included in the space cost?
In particular, the issue of whether there are general compilers that translate from one model to another comes up.

To better understand the differences between all these models, \cite{stoeckl2023streaming} studied the missing item finding problem (MIF): Given a data stream $e_1, \ldots, e_\ell$ of elements in $\{1,2,\dots,n\}$, possibly with repetitions, the goal of $MIF(n, \ell)$ problem is to find some $e \in \{1,2,\dots,n\}$ which has not appeared in the stream so far. I.e.\ $e \neq e_i$ for any $i \in \{1,2,\ldots,\ell\}$.  
This problem is of interest as it is a simple and natural search problem, that exhibits significantly different space complexities for classical
randomized algorithms, adversarially robust algorithms, deterministic algorithms and, as we will see, also in the "random start model", where all the random bits that have been used by the streaming algorithm are included in the space cost. Some of Stoeckl's results are listed in Table \ref{table:stoeckl-results}.

This work tackles some open questions and explores uncharted territories of the Missing Item Problem. We make the following contributions:  
\begin{itemize}
\item In Section~\ref{subsec:length-of-stream-close-n} we study the classical model, where the stream is fixed in advance. Stoeckl showed that the space complexity of randomized algorithms for $MIF(n,\ell)$ with $\ell \leq n-n/polylog(n)$ and $\delta = 1/poly(n)$ is $polylog(n)$. However, when the length of the stream $\ell$ is close to $n$ (say $\ell = n-\sqrt{n}$), the upper bound becomes polynomial in $n$ (say $\sqrt{n}$), while the lower bound becomes trivial (say $\Omega(1)$), which yields a large gap between them. For this range we show an upper bound of roughly $O(\log^2(n))$ for the space complexity. 

\item In Section \ref{subsec:mif-long-regime}, we study the $MIF(n,\ell)$ problem in the long regime where the length of the stream $\ell$ exceeds the size of the universe $n$. For $\ell = n+k$, we show a tight bound for the space complexity of roughly $\tilde{\Theta}(k)$.
        
\item In Section \ref{subsec:pseudo-deterministic-random- start}, we study the $MIF(n,\ell)$ problem in an adversarial setting, where the streaming algorithm should provide correct answers (with high probability) even when the stream updates are chosen by an adversary who may observe and react to the past outputs of the algorithm. First, we show a lower bound for a pseudo-deterministic error-zero algorithm of roughly $\Tilde{\Omega}(\ell)$. Based on the work of Stockel, we show how to convert this result into a lower bound for a random-start error-zero streaming algorithm of $\Tilde{\Omega}(\sqrt{\ell})$. 
Finally, we study the case of streaming algorithms with randomness-on-the-fly, where random bits that are saved for future use are included in the space cost. In this setting, for a stream with length $\ell = o(\sqrt{n})$, We show an upper bound of $O(\log n)$.
\end{itemize}

\begin{table}[h!]
\begin{center}
\begin{tabular}{ | p{6em} | p{6cm} | p{6cm}| } 
Model & Lower Bound & Upper Bound \\
 \hline \hline
 Classical & $\Omega\left(\sqrt{\frac{\log(1/\delta)}{\log(n)}}+\frac{\log(1/\delta)}{\log(n)(1+\log(n/\ell))}\right)$ & $min(\ell, \frac{\log(1/\delta)}{\log(n/\ell)})$  \\ 
 \hline
 Adv.\ Robust & $\Theta(\ell^2/n+\log(1-\delta))$ & $O\left(min(\ell, \left(1+\frac{\ell^2}{n}+\ln(\frac{1}{\delta})\right)\cdot \log \ell)\right)$ \\ 
 \hline
 Deterministic & $\Theta\left(\sqrt{\ell} + \frac{\ell}{1+\log(n/\ell)} \right)$ & $O\left(\sqrt{\ell \log \ell} + \frac{\ell \log \ell}{\log n} \right)$ \\ 
 \hline
  Random start & \text{Conditional lower bound of} $\Theta(\frac{\sqrt{\ell}}{poly\log n})$ \tablefootnote{Assuming lower bound for pseudo-deterministic algorithms of
$\Omega(\ell/ polylog n)$ bits of memory.}  & $O\left((\sqrt{\ell}+\ell^2/n)\log n\right)$ \\ 
\hline
\end{tabular}
\caption{Some of \cite{stoeckl2023streaming} results.}
\label{table:stoeckl-results}
\end{center}
\end{table}

In a parallel work, \cite{chakrabarti2023random} showed a similar lower bound for pseudo deterministic algorithm for the MIF problem, which implies a similar lower bound as ours for the random start model. Their results are more general than in this work since it holds also in the non-error-zero model. On the other hand, their proof is much more complicated than ours.

\section{Related Problems}\label{sec:related-problem}

\textbf{Adversarial streams:}  A streaming algorithm is called adversarially robust if its performance is guaranteed to hold even when the elements in the stream are chosen by an adaptive adversary, possibly as a function of previous estimates given by the streaming algorithm. Recently,  there has been a growing interest in studying the question of whether tasks that can be done with low memory against a stream that is fixed in advance can be done with low memory even again adaptive adversary.    

On the positive side, \cite{ben2022framework} showed general transformations for a family of tasks, for turning a streaming algorithm to be adversarially robust (with some overhead).

On the negative side,  \cite{kaplan2021separating} showed a problem that requires only a polylogarithmic amount of memory in the static case but any adversarially robust algorithm for it requires a polynomial space. \cite{chakrabarti2021adversarially} study the problem of maintaining a valid vertex coloring of an $n$-vertex graph with maximum degree $d$ given a stream of edges. In the standard, non-robust, streaming setting, ($d+1$)-colorings can be obtained while using only $\Tilde{O}(n)$ space, but any adversarially robust algorithm requires a linear amount of space, namely $\Omega(nd)$.

\noindent
\textbf{Card Guessing:} The card guessing, introduced by \cite{menuhin2021keep}, is a game between two players, Guesser and Dealer. At the beginning of the game, the Dealer holds a deck of $n$ cards (labeled $1,...,n$). For $n$ turns, the Dealer draws a card from the deck, the Guesser guesses which card was drawn, and then the card is discarded from the deck. The Guesser receives a point for each correctly guessed card, and its goal is to maximize its expected score. One can think about card guessing as a streaming problem where the algorithm tries to predict the next element in a stream or alternatively tries to predict an element that has not appeared in the stream yet. The promise is that all the elements in the stream are \textbf{unique}, which is the main difference from MIF and makes it an easier problem. One more difference is that in Card Guessing, the Guesser is allowed to make an incorrect guess, whereas in MIF the algorithm must report a missing item in all of the steps. 

They showed a probabilistic Guesser that uses $O(\log(1/\delta)\log^2 n)$ memory bits that expected to score $(1-\delta)\ln n$ correct guesses in a game against any static
Dealer (i.e.\ the stream is fixed in advance), by following the items that appeared in various subsets of $[n]$, and reconstructing the {\em unique} missing item in {\em some} subset. They showed a corresponding lower bound of $\Omega(\log^2 n)$. In contrast, for the adversarial case they showed that for every $m$ there exists an adaptive Dealer against which any
Guesser with $m$ bits of memory can score in expectation at most $\ln m + 2 \ln \log n + O(1)$ correct guesses.

\noindent
\textbf{Mirror Games:} In the Mirror Games \cite{garg2018space}, Alice and Bob take turns (with Alice playing first) in declaring numbers from the set $\{1,2, \ldots, 2n\}$. If a player picks a number that was previously played, that player loses the game and the other player wins. If all numbers are
declared without repetition, the result is a draw. Bob has a simple mirror strategy that assures he won't lose the game and requires no memory. On the other hand, Garg and Schneider showed that every deterministic Alice
requires memory of size that is proportional to $n$ in order to secure a draw. Feige~\cite{feige2019randomized} showed a randomized strategy for Alice that manages to draw against any strategy of Bob with probability at least $1-\frac{1}{n}$ and requires $O(\log^3 n)$ bits of memory. \cite{magen2023mirror} showed that for any white box Alice (Bob knows Alice's memory but not her future coin flips) and at most $n/4c$ bits of memory, there is a Bob that wins with probability close to $1-2^{-c/2}$.

\section{Preliminaries}
\textbf{Notations:} Let $[n]$ be the set $\{1,2,\dots,n\}$.  We use standard big-Oh notation, with  $\Theta(\cdot), \Omega(\cdot), O(\cdot)$ hiding constants and $\Tilde{\Theta}(\cdot), \Tilde{\Omega}(\cdot),\Tilde{O}(\cdot)$ hiding constants and factors that are polylogarithmic in the problem parameters. 

\subsection{Streaming Algorithms and Models}
A data stream of length $\ell$ over a domain $[n]$ is a sequence of updates of the form
$$(e_1,\Delta_1),(e_2,\Delta_2),\dots,(e_m,\Delta_m)$$ 
where $e_i \in [n]$ is a streaming item and $\Delta_j \in \Z$ is an increment or decrement to that item. Observe that in $MIF$, we have that $\Delta_j = 1$ for any $j \in [\ell]$. 
The frequency vector $f \in \R^n$ of the stream is the vector with $i^{th}$ coordinate $f_i = \sum_{j\leq \ell:e_j = i} \Delta_j$. In the insertion-only model, the updates are assumed to be positive, meaning $\Delta_j > 0$, whereas
in the turnstile model $\Delta_j$ can be positive or negative.

\noindent
\textbf{Perfect $L_p$ Samplers:} A perfect $L_p$ sampler is a sampling algorithm from a stream that should output an item $i\in[n]$ with probability  $|f_i|^p  / \norm{f}_p^p$, and is allowed to fail with a certain probability. We will be using a construction of ~\citet{jayaram2021perfect} for our purposes. 

\begin{definition}[Pseudo-deterministic Streaming Algorithm] \label{def:pseudo-deterministic} A pseudo-deterministic streaming algorithm is a (randomized) streaming algorithm $A$ such that for all valid input streams $\sigma =
e_1, \ldots , e_\ell \in [n]^\ell$, there exists a corresponding $e_\sigma \in \R$ for which  
\begin{align*}
    \Pr_r [\A(\sigma,r) = e_\sigma] \geq \frac{2}{3}.
\end{align*}
\end{definition}

\begin{definition}[Zero Error Algorithm] \label{def:zero-error} A zero error algorithm is a randomized algorithm that with probability $1$ outputs a correct answer or $FAIL$. 
\end{definition}
\begin{definition}[Random Start Algorithm] \label{def:random-start}
A random start algorithm is a deterministic algorithm with an initial random state. In this model, all random bits used are included in the space cost. The algorithm only has access to the randomness it had when it started.
\end{definition}

\begin{definition}[Streaming Algorithm with randomness on the fly] \label{def:randomness_on_the_fly} An algorithm with randomness on the fly is a randomized algorithm with access to new random bits in each turn but still has limited memory. Namely, random bits that should read at the same time or random bits that are saved for future use are included in the space cost.
\end{definition}
\subsection{Communication Complexity}
Consider a predicate $f\colon \X \times \Y \rightarrow \Z$ where we assume in the typical case that $\X=\Y=\{0,1\}^n$ and $Z=\{0,1\}$. Suppose Alice’s input is an element $x$ from $\X$ and Bob’s input is an element $y$ from $\Y$. A communication protocol is an algorithm to generate a conversation between Alice and Bob to compute $f(x,y)$. The total number of bits sent from Alice to Bob and vice versa is the complexity of the protocol. A one-way communication protocol is a communication protocol where only Alice can send a message to Bob, and then Bob needs to report his answer.   

The celebrated  theorem regarding  the  communication complexity of disjointness as phrased in \citep{rao2020communication,goos2016communication,braverman2013information}) ia:
\begin{theorem}[Theorem 6.19 in \cite{rao2020communication}] \label{thm-cc-disjointness-lower-bound}
    Any randomized communication protocol that computes disjointness function with error $1/2-\epsilon$ must have communication $\Omega(\epsilon n)$.
\end{theorem}
\textbf{The avoid(t, a, b) communication task} This one-way communication game was introduced by \cite{chakrabarti2021adversarially}. In it, Alice is given a set $S_A \subset [t]$ with $|S_A| = a$, and sends a message to Bob, who must produce a set
$S_B \subseteq [t]$ with $|S_B| = b$ where $S_B$ is disjoint from $S_A$.

\begin{lemma}[Lemma 4.1, from \cite{chakrabarti2021adversarially}] \label{lemma:avoid-lower-bound}
    The public-coin $\delta$ error one-way communication complexity
of $avoid(t,a,b)$ is at least $\frac{ab}{t ln 2} + \log(1-\delta)$
\end{lemma}

\section{Classical Model: Missing Item with a Long Stream} \label{sec:classical-model-long-regime}
In this section, we study the classical randomized model, where the stream is fixed in advance. We focus on long streams, where the length of the stream $\ell$ is close to the size of the universe $n$ or even bigger than it.

\subsection{When the stream  length is close to the size of the universe} \label{subsec:length-of-stream-close-n}

We begin our exploration by studying $MIF(n,\ell)$ when $\ell < n$. In this range, \cite{stoeckl2023streaming} showed an upper bound of $min(\ell, \frac{\log(1/\delta)}{\log(n/\ell)})$, where $\delta$ is the error. This is achieved by an algorithm that maintains a list of random items from $[n]$, hoping to find an unused number at the end of the stream. When $\ell = n-k$, the size of that list should be roughly $n/k$. Indeed, where $k$ is sublinear in $n$ (say $k = \sqrt{n}$) we have 
   On the one hand, we have that $\log(n/(n-k)) = -\log(1-k/n)$. By L'Hopital's rule, we have that $\lim_{x\rightarrow 0} \log(1+x)/x = 1$. Therefore, $-\log(1-k/n)$
goes to zero as $k/n$, this implies that Stoeckl's upper bound scale as $n/k$, which becomes polynomial in $n$, where $k \leq n^c$ for some constant $c<1$.
On the other hand, Stoeckl showed a lower bound of $\Omega\left(\sqrt{\frac{\log(1/\delta)}{\log(n)}}+\frac{\log(1/\delta)}{\log(n)(1+\log(n/\ell))}\right)$.  
Observe that $1+\log(n/\ell) = 1-\log(1-k/n)$ goes to $1$, when $n$ goes to infinity. This implies a trivial lower bound of $\Omega(1)$, for constant error.

We aim to close this gap. We will be using $L_1$ samplers from a stream, i.e.\ ones that sample an item with probability proportional to the t frequency of an item in $L_1$ norm. We apply the following result about low space $L_1$ samplers:

\begin{theorem}[Thm.~9 from \citet{jayaram2021perfect}] \label{thm: perfect-l1-sampler}
There is an algorithm $\A$ which, on a general turnstile stream $f$, outputs $i \in [n]$ with probability $(1 \pm v)|f_i| / \norm{f}_1 + O(n^{-c})$ for some constant $c \geq 1$, and outputs FAIL with probability at most $\delta_1$. Conditioned on outputting some $i \in [n]$, $\A$ will then output $\Tilde{f}_i$ (the estimation for $f_i$) such that $\Tilde{f}_i = (1 \pm \epsilon)f_i$ with probability $1-\delta_2$. The space required is $$O\bigl( \left(\log^2 n (\log\log n)^2+\beta \log n \log(1/\delta_2)\right)\log(1/\delta_1) \bigr),$$ where $\beta = \min \{\epsilon^{-2}, \max \{\epsilon^{-1}, \log(1/\delta_2)\} \}$.  
\end{theorem}

\begin{algorithm}
\caption{A streaming algorithm for $MIF(n)$ with constant error on any input stream}\label{alg:mif-static}
\begin{algorithmic}
\State \textbf{Parameters:} $\delta_2$ 
 \State \textbf{Initialization:}
 \State Let $\A$ be $L_1$-sampler of Theorem \ref{thm: perfect-l1-sampler} with $\epsilon = v = \delta_1 =1/4, \delta_2$.
 \State Feed the updates $(i,-1)$ for $i = 1, \dots, n$ to $\A$. \\ 
 \State \textbf{Update}($i \in [n]$):
 \State Feed the update $(i,1)$ to $\A$ \\

 \State  \textbf{Query:}
  \If {$\A$ outputs some $i \in [n]$ with negative $\Tilde{f}$}
  \Comment{$\Tilde{f}$ is  the estimation of $f_i$}
  \State output $i$
 \Else 
 \State output FAIL
 \EndIf
 
 \end{algorithmic}
\end{algorithm}

Now, we are ready to show an upper-bound on the space complexity of $MIF(n,\ell)$ with $
\ell<n$:  
\begin{theorem} \label{thm:mif-static}
For any $\delta > 0$ there is a $O(\log^2 n\log\log n \log^2(1/\delta))$ space one-pass algorithm which, given a stream
of length $\ell$ over the universe $[n]$ with $\ell<n$, outputs missing item $i \in [n]$ with probability at least $1-\delta$.
\end{theorem}

\begin{proof}
Let $\delta > 0$ and let $\B$ be the algorithm that repeats algorithm \ref{alg:mif-static} for $O\left(\log(1/\delta)\right)$ times in parallel and reports the first non-failing output. 

First, we argue that Algorithm \ref{alg:mif-static} outputs FAIL with constant probability (since $\delta_1 = 1/4$). Second, given that the algorithm outputs $i$ with negative $f_i$, then $i$ is an item that has not appeared in the stream. Indeed, in the initialization phase, we subtract $1$ from each coordinate of $f$. When a stream item $i \in [n]$ comes, we increase $f_i$ by $1$. Therefore, we have that $f_i = -1$ for items that have not appeared in the stream, $f_i = 0$ for items that appeared once, and $f_i > 1$ for items that appeared at least twice. At the end of the initialization phase, we have that $\sum_{i=1}^{n} f_i = -n$.

Since the length of the stream $\ell$ is smaller than $n$, we have that at any point in the stream $\sum_{i=1}^{n} f_i < 0$. Therefore, a perfect $L_1$ sampler for $f$ outputs $i$ such that $f_i$ is negative with probability of at least $1/2$. Moreover, since $v=1/4$ (the parameter that effects the estimation of $\A$ - the $L_1$ sampler from Thm. \ref{thm: perfect-l1-sampler}), the probability that $\A$ outputs $i$ such that $f_i$ is negative, given that $\A$ not outputs FAIL, is at least $1/3$. To conclude, the overall Algorithm \ref{alg:mif-static} outputs $i$ such that $f_i$ is negative with probability at least $\frac{1}{3}\cdot\frac{3}{4} = \frac{1}{4}$.

To amplify the probability to $1-\delta$, Algorithm $\B$ repeats this process $O\left(\log(1/\delta)\right)$ times in parallel and report the first item $i$ with negative $f_i$. However, we have just an estimation for $f_i$ and hence need to ensure (with high probability) that the estimation $\Tilde{f}$ has the same sign as $f_i$ throughout the amplification process. Choosing $\delta_2 = O\left(\delta / \log(\delta^{-1})\right)$, we have by the union bound that $\Tilde{f}\cdot f_i > 0$ for all the $O\left(\log(1/\delta)\right)$ instances of the algorithm with probability at least $1-\delta/2$. Overall, algorithm $\B$ solves $MIF(n,\ell)$ for $\ell < n$ with probability $1-\delta$. Since the space complexity of algorithm \ref{alg:mif-static} with $\delta_2 = O\left(\delta / \log(\delta^{-1})\right)$ is $O\left( \log^2 n\log\log n \log(1/\delta) \right)$ bits of memory, the result follows.

\end{proof}

\begin{remark}
    The $L_1$ sampler In Thm.~\ref{thm:mif-static} can be replaced with every $L_P$ sampler, where $0\leq p \leq 1$. To see this, need to observe that after the initialization phase $f_i = -1$ for any $i\in[n]$, and hence at the end of the stream $f_i=-1$ for every unseen item $i\in [n]$. Moreover, given stream $e_1,\dots, e_{\ell}$ if $A = \{i\in[n]: \forall j\in[\ell], i \neq e_{j}\}$ is the set of
    unseen items, and $B = \{i\in[n]: \exists j,j'\in[\ell], j\neq j'. i = e_{j}=e_{j'}\}$ is the set of items that have appeared at least twice in the stream, then we have
    \begin{align*}
        \sum_{i \in A} |f_i|^p = |A| \geq \sum_{i \in B} f_i \geq \sum_{i \in B} |f_i|^p,
    \end{align*}
    for $0\leq p \leq 1$. Therefore, a perfect $L_p$ sampler will sample an item $i$ from $A$ with a probability of at least half. The rest of the arguments are the same as in the proof.  
\end{remark}

\noindent
\textbf{Application to Card Guessing:}
Recall the Card Guessing game introduced in Section \ref{sec:related-problem}. Using the same algorithm as in Thm.~\ref{thm:mif-static} with $\delta = \delta'/n$, we can construct a Guesser that uses $O(\log(n/\delta')\log^2(n))$ memory bits that are expected to score $(1-\delta')\ln(n)$ correct guesses against any static Dealer (i.e.\ the deck is fixed in advanced). Although it does not improve \cite{menuhin2021keep} result, it shows the versatility of the $L_1$ sampler (note that this observation was made by Arnold Filtser (private communication) as well). 

\subsection{Missing Item problem In The Long Regime} \label{subsec:mif-long-regime}
After showing that the space complexity of $MIF(n,\ell)$ is $polylog(n)$, even when the length of the stream is close to $n$, it is natural to ask if it is possible to achieve a similar result when the total length of the stream is longer than the size of the domain (i.e. $\ell \geq n$). In this case, we say that algorithm $\A$ solves $MIF(n, \ell)$ with error $\delta$ if and only if $\A$ outputs a missing item $i \in [n]$ at the end of the stream (if such an item exists) with a probability of at least $1 - \delta$. If all the items have already appeared in the stream, we allow $\A$ to output any number. In the next part, we show that perhaps surprisingly, the answer is negative.

\begin{theorem} 
For any $\delta \in [0,1/4]$, the space complexity for an algorithm solving $MIF(n, \ell)$ with $\ell > n+k$ (where $k\leq n$ is a function of $n$) and error $\delta$ is at least $\Omega(k)$. 
\end{theorem}
\begin{proof}
We prove this by reducing the communication task of disjointness to $MIF(n, \ell)$ with $\ell \geq n+k$. Suppose that the two parties Alice and Bob have subsets $S_A \subseteq [k]$ and $S_B \subseteq [k]$ respectively and the question is whether $S_A \cap S_B =\emptyset$. They instantiate an instance $\X$ of the given algorithm for $MIF(n,\ell)$, then Alice runs it on $\comp{S_A}$ and $[k+1,n]$, where $\comp{S_A}$ is the complement of $S_A$ in $[k]$. Then she sends the state of $\X$ to Bob. Since this is a public coin protocol, all randomness can be shared for free. Then Bob runs $\X$ on $\comp{S_B}$. 
Overall, they run $\X$ on at most $|\comp{S_B}|+|\comp{S_A}| + |[k+1,n]| = n+k+1$ items. In the end, $\X$ outputs a candidate $x \in [n]$ for a missing item. If $x \in [k]$, then Bob sends $x$ to Alice, and both of them verify that $x$ belongs to their subset $S_A$ and $S_B$. If this is indeed the case, then Bob outputs $1$, which means that $S_A$ intersects with $S_B$. Otherwise, Bob outputs $0$,  which means that $S_A$ is disjoint to $S_B$.

Indeed, if $S_A \cap S_B \neq \emptyset$, then $\comp{S_A} \cup \comp{S_B} \cup [k+1,n] \neq [n]$, which means that with probability at least $1 - \delta$ we have that $x \not \in \comp{S_A} \cup \comp{S_B} \cup [k+1,n]$, which implies that $x \in S_A \cap S_B$. Therefore, with probability at least $1 - \delta$ the protocol is correct. Otherwise, $S_A \cap S_B = \emptyset$,  and hence for any $x \in [k]$, we have that $x \notin S_A \cap S_B$, which implies that the protocol is correct with probability 1. Overall, this protocol uses at most $s + (\log k) + 2$ bits, where $s$ is the space of $\X$.  Combine with Thm.~\ref{thm-cc-disjointness-lower-bound}, the result follows.
\end{proof}

We now show that the bound above is nearly tight by using the same ideas of Thm.~\ref{thm:mif-static}, with a slightly different analysis. 
\begin{theorem} \label{thm:mif-static}
For any $\delta > 0$ there is a $O(k\log^2 n\log\log n\log(1/\delta))$ space one-pass algorithm which, given a stream
of length $\ell$ over the universe $[n]$ with $\ell=n+k$, outputs a missing item $i \in [n]$ with probability at least $1-\delta$.
\end{theorem}
\begin{proof}
As in the proof of Thm.~\ref{alg:mif-static}, at the end of the stream we have that $f_i = -1$ for items that don't appear in the stream, $f_i = 0$ for items that appear once, and $f_i > 1$ for items that appear at least twice.
At the end of the initialization phase we have that $\sum_{i=1}^n f_i = -n$, which implies after $n+k$ stream items that $\sum_{i=1}^n f_i = k$.

Assuming that there is indeed a missing item in the stream, 
then a perfect $L_1$ sampler for $f$ outputs $i$ such that $f_i$ is negative with a probability of at least $1/k$.
Hence, by repeating algorithm \ref{alg:mif-static} with $\delta_2 = 0(\delta / k\log(1/\delta))$ for $O(k\log(1/\delta))$  times in parallel and accepting the first non-failing output, the results follows.
\end{proof}

\section{Lower Bounds for Pseudo-deterministic and Random Start Models} \label{subsec:pseudo-deterministic-random- start}

This section studies the $MIF(n,\ell)$ in an adversarial setting, where the adversary might choose the stream elements based on the past streaming algorithm's outputs.
We start by showing a tight lower bound of $\Tilde{\Theta(\ell)}$ (up to $polylog(n)$ factors) for a pseudo-deterministic error-zero algorithm. We note that Stoeckl showed a tight bound of $O(\ell)$ (up to $polylog(n)$ factors) for a deterministic algorithm and that the next lower bound is based on some of his ideas. In more detail, given partial stream $\sigma \in [n]^*$, let $F(\sigma)$ be the set of possible outputs of the deterministic algorithm when $\sigma$ extends to a full stream. The key idea in Stoeckl's proof is to show that \\ \\
\textbf{Key Idea:} Exists partial stream $\sigma^* \in [n]^t$ such that for every partial stream $a \in [n]^k$ (for the right choice of $k$ and $t$) we have that $|F(\sigma^*a)|$ is large compare to $|F(\sigma^*a)|$ (say  $|F(\sigma^*a)| \geq |F(\sigma^*)|/2$). \\ \\ 

Such a partial stream exists by a recursive method that starts with the empty stream $\sigma_0$. If the empty stream has the above property, we are done. Otherwise, exists partial stream $\sigma_1 \in [n]^k$ such that $|F(\sigma_0\sigma_1)| < |F(\sigma_0)|/2 \leq n/2$. If $\sigma_0\sigma_1$ has that property, we are done. Otherwise, continue in that process and after $\log(n)$ times, this process must end, otherwise $|F(\sigma_0\sigma_1\dots\sigma_n)| \leq 0$, which is a contradiction to the correctness of the deterministic algorithm. 

Then, using the deterministic algorithm for MIF we can construct an algorithm for avoid($t = F(\sigma^*)|, a=k , b = F(\sigma^*)|/2$): First feed the deterministic algorithm with $\sigma^*$, then with $S_A$, which is Alice's input at size $k$ in the avoid problem. By the property of $\sigma^*$, we have that the set of the possible outputs of the algorithm when we extend $\sigma^*S_a$ to a full stream is large (at least $F(\sigma^*)|/2$), and all these possible outputs are not in $S_A \cup \sigma^*$. Since the algorithm for MIF is deterministic, we can run it on every possible extension of $\sigma^*S_A$ to a full stream to reconstruct a large set $S_B$ which is disjoint from $S_A$. Therefore, a deterministic lower bound for the avoid problem can be translated to a deterministic lower bound for MIF problem. 

To adapt these ideas to a lower bound for a pseudo-deterministic algorithm, we first need to define $F(\sigma)$ in a different way, since $F(\sigma)$ as defined above is not necessarily disjoint from the elements in $\sigma$. The new definition of $F(\sigma)$ is based on the outputs with a high probability of a pseudo-deterministic algorithm. Then, we show the existence of a partial stream $\sigma^* \in [n]^*$ with the same property as in the key idea of Stoeckl. To reduce the avoid problem to the MIF problem, we used the assumption of the zero error model, which implies that after we run the algorithm on any partial stream that extends $\sigma^*S_A$ (where $S_A$ is the input of Alice in the avoid problem) to a full stream, the set of items we get is disjoint from $\sigma^* \cup S_A$. Ultimately, we carefully use a union bound argument to show that this set is large enough. Formally, we have the next Theorem:

\begin{theorem} \label{thm:pseudo-det-lower-bound}
Every pseudo-deterministic (Definition~\ref{def:pseudo-deterministic}) error-zero (Definition~\ref{def:zero-error}) streaming algorithm for $MIF(n,\ell)$ with error at most $\delta \leq 1/3$ requires $\Omega(\ell / \log^2n+\log(1-\delta)/\log n)$ bits of space.
\end{theorem}
 \begin{proof}
 We prove this by reducing the communication task of $avoid$ to the missing item problem. Let $\A'$ be a pseudo-deterministic zero error streaming algorithm for $MIF(n,\ell)$ with error $\delta \leq 1/3$. 
 Recall that by the definition of pseudo-determinism, for every stream $\sigma = e_1,\dots,e_{\ell} \in [n]^\ell$, there exists an item $e_\sigma \in [n]$ such that
 \begin{align*}
     \Pr_r [\A(\sigma,r) = e_\sigma] \geq 2/3
 \end{align*}
Since the probability of failure $\delta \leq 1/3$, we have that $e_\sigma \neq e_i$ for any $i \in [\ell]$.
If the space complexity of $\A$ is $s$, then  by repeating the algorithm $\Theta(\log n)$
times in parallel and outputting the first non-failing output, we get an algorithm $\A$ with $\Theta(s \cdot \log n)$ bits of space that has
\begin{align} \label{eq:amlify-zero-error-pseudo-deterministic} 
 \Pr_r [\A(\sigma,r) = e_\sigma] \geq 1-\frac{\delta}{n} 
\end{align}
For any partial stream $\sigma \in [n]^{\ell'}$ with $\ell' \leq \ell$, define
\begin{align*}
    \F(\sigma) = \{e_{\sigma\sigma'} : \sigma' \in [n]^{\ell -\ell'}  \}
\end{align*} 
to be the set of possible outputs of $\A$ when $\sigma$ extends to a full stream (i.e. to a stream of length $\ell$). Let $k = \ceil{\ell / 2\log n}$.
We argue that there exists a partial stream $\sigma^* \in [n]^*$, such that for every partial stream $a \in [n]^k$, we have that

\begin{align}\label{eq:f-sigma-possible-outputs-remains-high}
    \left|\F(\sigma^* a)\right| \geq \frac{1}{2}\left|\F(\sigma^*) \right|.
\end{align}

To see this, we will construct a (partial) stream $\sigma^*$ with this property. Denote the empty stream by $\sigma^0$. In step number $i$ (starting with $i=0$), if there exists a string $a \in [n]^k$ such that $\left|\F(\sigma^ia)\right| \leq \frac{1}{2}\left|\F(\sigma^i) \right|$, then let $\sigma^{i+1} = \sigma^i a$. Observe that until step $i = \ceil{\log n}$, we have that $|\sigma^i| = ki \leq \ell / 2$, and thus $\F(\sigma^i)$ is well defined. 
 Moreover, since $\left|\F(\sigma^0)\right| \leq n$, the above process must end after at most $i_{max}:=\ceil{\log(n/\ell)+2}$ steps. Otherwise $|\F(\sigma^{i_{max}})| \leq \ell/4$, which implies that $\A$ cannot output a correct answer on the stream $\sigma^{i_{max}}\sigma'$, where $\sigma'$ is a partial stream that contains all the elements of $\F(\sigma^{i_{max}})$ and extends $\sigma^{i_{max}}$ to a full stream of length $\ell$ (i.e. $|\sigma\sigma'| = \ell$). We conclude that  there exists a stream  $\sigma^* \in [n]^*$ satisfying Eq.~(\ref{eq:f-sigma-possible-outputs-remains-high}) for every $a\in [n]^k$.

We are now ready to construct a random public coins protocol for  $avoid(|\F(\sigma^*)|,k, \frac{1}{2}|\F(\sigma^*)|)$.
Suppose that Alice has as input a subset  $S_A \subseteq \F(\sigma^*)$ with $|S_A| = k$ and Bob needs to produce a set $S_B \subseteq \F(\sigma^*)$ with $|S_B| = \frac{1}{2}|\F(\sigma^*)|$ such that $S_A \cap S_B = \emptyset$. They instantiate an instance $\A$ of the given pseudo-deterministic zero error algorithm for $MIF(n,\ell)$. Then Alice runs it on $S_A$ and sends the state of $\X$ to Bob. Since this is a public coins
protocol, all randomness $r$ can be shared for free. Then Bob runs $\A$ on every possible partial stream $\sigma'\in [n]^{\ell -k - \ell^*}$, where $\ell^* = |\sigma^*|$ and reports $S_B$ as the first $\frac{1}{2}|\F(\sigma)|$ items that $\A$ outputs (if there are sufficiently many   such elements).  Let $a$ be the partial stream that consists of all the items in $S_A$ in some arbitrary order, and let
\begin{align}
    G_r(a) = \{e \in [n] : A(\sigma^*a\sigma',r) = e, \sigma' \in [n]^{\ell -k-\ell'}\}
\end{align}
be the set of all the outputs of $\A$ on all extensions of $\sigma^*a$ to a full stream, given randomness $r$. Since $\A$ is a zero error algorithm, we have that $G_r(a) \subseteq \F(\sigma^*a)$ for every possible randomness $r$. Moreover, if $\F(\sigma^*a) = \{e_1,\dots,e_t\}$, for $1/2|\F(\sigma^*)| \leq t \leq n$, then for any $e_i\in \F(\sigma^*a)$ there exists some $\sigma_i \in [n]^{\ell - k - \ell'}$  such that $e_i$ is the unique possible output of $\A$ on the stream $\sigma^*a\sigma_i$. By Eq.~(\ref{eq:amlify-zero-error-pseudo-deterministic}) and the union bound, we have that 
\begin{align*}
   Pr_r[\forall i\in [t], \ \A(\sigma^*a\sigma_i) = e_i] \geq 1-\delta.
\end{align*}
Namely, $G_r(a) \supseteq \F(\sigma^*a)$ with probability of at least $1-\delta$, which implies a protocol that solves $$avoid(|\F(\sigma^*)|,k, \frac{1}{2}|\F(\sigma^*)|)$$ with probability of at least $1-\delta$.

Observe that this protocol uses $\Theta(s \log n)$ bits of memory, where $s$, as defined at the beginning of the proof,  is the space complexity of $\A'$. By Lemma~\ref{lemma:avoid-lower-bound}, we have that 
\begin{align*}
    \Theta(s \log n) \geq \frac{k}{2log(2)} + \log(1-\delta) \geq \frac{\ell}{4\log n} + \log(1-\delta),
\end{align*} 
from which the result follows. 
 \end{proof}

Observe that this lower bound is tight up to $polylog(n)$,  since a deterministic algorithm that remembers all the unused items from $[\ell+1]$ can be implemented using $\ell+1$ bits of memory.

We now consider the random start model, where the random bits used by the streaming algorithm are included in the space (memory) cost.
\cite{stoeckl2023streaming} showed a {\em conditional} lower bound on the space complexity of a random start algorithm for the missing item problem, which depended on the space needed for a pseudo-deterministic algorithm to solve that problem. Using almost identical proof, we show that such a conditional lower-bound also holds in the zero error case. Combining with our lower bound for pseudo-deterministic error-zero (Thm.~\ref{thm:pseudo-det-lower-bound}), we have the next result which is given mainly for completness:  
 \begin{theorem} \label{thm:lower-bound-random-start}
Every random start (Definition \ref{def:random-start}) error zero (Definition \ref{def:zero-error}) streaming algorithm for $MIF(n,\ell)$ against adaptive adversaries, with error at most $\delta \leq 1/3$ requires $\Omega(\sqrt{\ell}/\log n)$ bits of space.
 \end{theorem}
 \begin{proof}
Let $\A$ be a random start error zero streaming algorithm for $MIF(n,\ell)$ that uses $s$ bits of space. Using $\A$, we will construct a pseudo-deterministic error zero streaming algorithm for $MIF(n,\ell/2s)$ that uses the same space as $\A$. 

Given a partial stream $\sigma \in [n]^*$, let $\A(\sigma,r)$ be the sequence of $|\sigma|$ outputs made by $\A$ given input $\sigma$ and randomness $r$. Let $k = \ceil{\ell / (2s+3)}$.
We argue that with high probability there exists a partial stream $\sigma_{r}^* \in [n]^*$ (where $r$ is the randomness of $A$) with $|\sigma_{r}^* \in [n]^*| \leq \ell - k$, that for every partial stream $a \in [n]^k$ there exists a corresponding output $b \in [n]^k$ such that
\begin{align} \label{eq:random-start-easy-to-predict}
    \Pr_r [A(\sigma_{r}^*a,r) = w b : A(\sigma_{r}^*,r) = w] \geq \frac{2}{3}.
\end{align}
We will show how to construct $\sigma_{r}^*$ satisfying this property. Denote the empty stream by $\sigma^0$. In step number $i$ (starting with $i=0$), if there exists $a \in [n]^k$ such that for every $b \in [n]^k$ we have that
\begin{align}\label{eq:random-start-hard-to-predict}
    \Pr_r [A(\sigma^ia,r) = w b : A(\sigma^i,r) = w] \leq \frac{2}{3},
\end{align}
then let $\sigma^{i+1} = \sigma^ia$. Otherwise, $\sigma_{r}^* = \sigma^i$ satisfies Eq. \ref{eq:random-start-easy-to-predict}. Assume that this construction doesn't end until turn $t := 2s+4$. Let $a_1,\dots,a_{t} \in [n]^k$ be the partial streams that satisfies Eq. \ref{eq:random-start-hard-to-predict}, and let $b_1, \dots, b_t \in [n]^t$ be the corresponding outputs of the algorithm. By
applying Eq. \ref{eq:random-start-hard-to-predict} repeatedly, we have that
\begin{align*}
    \Pr_r [A(a_1\dots a_t,r) = b_1\dots b_t] \nonumber &= \\ 
    &\Pr_r [A(a_1,r) = b_1] \cdot \Pr_r [A(a_1a_2,r) = b_1b_2: A(a_1,r) = b_1] \cdot \nonumber\\ 
    &\dots \nonumber\\
    &\cdot \Pr_r [A(a_1\dots a_t,r) = b_1\dots b_t : A(a_1\dots a_{t-1},r) = b_1\dots b_{t-1}] \nonumber\\
    &\leq \left(\frac{2}{3}\right)^t = \left(\frac{2}{3}\right)^{2s+4} \leq \left(\frac{1}{2}\right)^{s+2} = 2^{-(s+2)}. \label{eq:pseudo-deterministic-small-probability-that-construction-doesn't-end}
\end{align*}
Therefore, if $r_A$ is the actual random bits of $\A$, then $\Pr_r[r=r_A] \leq 2^{-(s+2)}$. Since the space of $\A$ is $s$ includes its random bits, we have that $|R| \leq 2^s$, where $R$ is the set of all possible random bits of $\A$. Let \begin{align*} 
    R_{low} = \left\{r' \in R : \Pr_r [r=r'] \leq 2^{-(s+2)}\right\} 
\end{align*}
be the set of all the random bits with a small probability. We have that  
\begin{align} 
    \Pr_r[r \in R_{low}] = \sum_{r' \in R_{low}} \Pr_r [r = r'] \leq |R_{low}| \cdot 2^{-(s+2)} \leq 2^s \cdot 2^{-(s+2)} = \frac{1}{4} .
\end{align}
Overall, we have that with probability of at least $3/4$, there exists a partial stream $\sigma_{r}^* \in [n]^*$ (where $r$ is the random bits of $A$) with $|\sigma_{r}^*| \leq (t-1)k \leq \ell - k$, that for every partial stream $a \in [n]^k$ there exists a corresponding output $b \in [n]^k$ that satisfying Eq. \ref{eq:random-start-easy-to-predict}, which is exactly the requirement for pseudo-deterministic Algorithm.  

We are finally ready to construct a pseudo-deterministic error-zero algorithm $B$ for the $MIF(n,k=\ceil{\ell / (2s+3)})$ problem. In the initialization phase, given random bits $r_A$ for $\A$, algorithm $B$ construct $\sigma_{r_A}^*$ that satisfies Eq.~(\ref{eq:random-start-easy-to-predict}). 
Observe that during this construction $B$ feeds $\A$ with the partial stream $\sigma_{r_A}^*$. Algorithm $B$ fails in that mission with a probability of at most $1/4$. Then given update stream $e\in [n]$, algorithm $B$ feed $\A$ with $e$, and outputs the same element that $\A$ outputs.
Since $\A$ is zero error with error at most $1/3$, we have that $B$ is a pseudo-deterministic error-zero algorithm for $MIF(n,\ceil{\ell / (2s+3)})$ with error at most $1/3 + 1/4 = 7/12$. By Thm.~\ref{thm:pseudo-det-lower-bound}, we have that:
\begin{align*}
    s \geq \frac{c\ell}{s\log^2 n} \Rightarrow s \geq \frac{c'\sqrt{\ell}}{\log n},
\end{align*}
where $c > 0$ is a constant and $c' = \sqrt{c}$.  
 \end{proof}
 \begin{remark}
     \cite{stoeckl2023streaming} established an upper bound for a random start error zero streaming algorithm that solves $MIF(n,\ell)$ of $O\left(\left(\sqrt{\ell} + \ell^2/n\right)\log n\right)$. Observe that $\ell^2/n \geq \sqrt{\ell}$ iff $\ell \geq n^{2/3}$, for that regime the above upper bound reduce to $O\left(\ell^2/n \cdot \log n\right)$, which means that the lower bound of adversarially robust from Table \ref{table:stoeckl-results} (which is also a lower bound for random start error zero streaming algorithm) is tight up to polylog factors. For $\ell \leq n^{2/3}$, the upper bound of Stoeckl reduces to $O\left(\sqrt{\ell}\log n\right)$, which means that in this case the lower bound of Thm. \ref{thm:lower-bound-random-start} is tight, up to polylog factors.
 \end{remark}
A natural question is whether an algorithm with small memory, which can ask for an unlimited number of new random bits, but cannot store all of them, can achieve a better performance as compared to the random start model.
The answer is positive, at least when the length of the stream is small (i.e. $\ell = o(\sqrt{n})$) with some errors based on a very simple algorithm:
 \begin{theorem}
     There is an algorithm $\A$ with randomness on the fly (Definition \ref{def:randomness_on_the_fly}) that solves $MIF(n, \ell )$ against adaptive adversaries, with error $ \frac{\ell^2 + \ell}{2n}$ and can be implemented using      
     $O(\log n)$ bits of space.
 \end{theorem}
 \begin{proof}
     Let $\A$ be an algorithm that in each turn chooses uniformly at random a number from $[n]$. In turn number $i \in [\ell]$ the probability that $\A$ chooses a number that has appeared in the stream so far is at most $i / n$. By the union bound, the probability of failure in the first $\ell$ turns is at most
     \begin{align*}
         \sum_{i=1}^\ell \frac{i}{n} = \frac{1}{n} \sum_{i=1}^\ell i = \frac{1}{n} \frac{\ell}{2} (\ell + 1) = \frac{\ell^2 + \ell}{2n}
     \end{align*}
 \end{proof}
 Observe that for $\ell = o(\sqrt{n})$ and constant error, we have an upper bound $O(\log n)$ on the space complexity for an algorithm with `randomness on-the-fly'. On the other hand, by Thm.~\ref{thm:lower-bound-random-start}, every random start error zero streaming algorithm  need at least $\Theta(\sqrt{\ell}/\log n)$ bits of space.
 
\bibliographystyle{plainnat}
\bibliography{missing_item_problem}
\end{document}